\documentclass[draftclsnofoot, onecolumn]{IEEEtran}
\usepackage{cite}
\ifCLASSINFOpdf
\else
\fi
  \usepackage[pdftex]{graphicx}

\usepackage[cmex10]{amsmath}
\usepackage{amsfonts}
\usepackage{amssymb,framed}
\usepackage{amsthm}

\newtheorem{thm}{Theorem}[section]
\newtheorem{cor}[thm]{Corollary}
\newtheorem{lem}[thm]{Lemma}
\newtheorem{pro}[thm]{Proposition}
\newtheorem{defn}[thm]{Definition}

\newtheorem{rmk}[thm]{Remark}

\DeclareMathOperator*{\supp}{supp}
\DeclareMathOperator*{\argmin}{arg\,min}

\usepackage{algorithm}
\usepackage{algorithmic}

\usepackage{enumerate}

\usepackage{color}

\newcommand{\my}[1]{{\color{black}{#1}}}
\newcommand{\mymy}[1]{{\color{black}{#1}}}
\newcommand{\myIII}[1]{{\color{black}{#1}}}

\begin{document}
% \title{Orthogonal Greedy Algorithm with Thresholding in Compressed
% Sensing}
\title{Orthogonal Matching Pursuit with Thresholding and its
  Application in Compressive Sensing}

\author{Mingrui~Yang,~\IEEEmembership{Member,~IEEE,}
        and~Frank~de~Hoog
        % <-this % stops a space
\thanks{Mingrui Yang is with Autonomous System Lab, Digital
  Productivity Flagship, CSIRO, Australia (email: mingrui.yang@csiro.au).}% <-this % stops a space
\thanks{Frank de Hoog is with Digital
  Productivity Flagship, CSIRO, Australia.}
% % <-this % stops a space
% \thanks{Manuscript received mmm dd, yyyy; revised mmm dd, yyyy.}
}

\maketitle

\begin{abstract}
Greed is good. However, the tighter you squeeze, the less you
have. In this paper, a
less greedy algorithm for sparse signal reconstruction in compressive sensing, named orthogonal matching pursuit with
thresholding is studied. Using the
global 2-coherence , which provides a ``bridge" between the well known mutual coherence and the restricted isometry constant, the performance
of orthogonal matching pursuit with thresholding is analyzed and more
general results for sparse signal reconstruction are obtained. It is also
shown that given the same assumption on the coherence index and the
restricted isometry constant as required for orthogonal matching
pursuit, the thresholding variation gives exactly the same
reconstruction performance \mymy{with significantly less complexity}. 

% In this paper we define a new coherence index, named \emph{2-coherence}, of a given
% dictionary and study its relationship with the traditional mutual
% coherence and the restricted isometry constant. By exploring this
% relationship, we obtain more general results on sparse signal reconstruction using
% greedy algorithms in the compressive sensing (CS) framework. In
% particular, we obtain an improved bound over the best known results
% on the restricted isometry constant for
% successful recovery of sparse signals using orthogonal matching pursuit (OMP).

% We also initialized a study of a thresholding type greedy algorithm named orthogonal matching
% pursuit with thresholding (OMPT), which is more feasible in practice than OMP. We
% analyze its performance in CS framework for both noiseless and noisy cases in
% terms of coherence indices and the restricted isometry constant. We show
% that given the same assumptions as required for OMP, it achieves exactly the same
% reconstruction performance as OMP.
\end{abstract}


% \fi
\IEEEpeerreviewmaketitle

\section{Introduction}
% % The very first letter is a 2 line initial drop letter followed
% % by the rest of the first word in caps.
% % 
% % form to use if the first word consists of a single letter:
% % \IEEEPARstart{A}{demo} file is ....
% % 
% % form to use if you need the single drop letter followed by
% % normal text (unknown if ever used by IEEE):
% % \IEEEPARstart{A}{}demo file is ....
% % 
% % Some journals put the first two words in caps:
% % \IEEEPARstart{T}{his demo} file is ....
% % 
% % Here we have the typical use of a "T" for an initial drop letter
% % and "HIS" in caps to complete the first word.
% for IEEE journal papers produced under \LaTeX\ using
% IEEEtran.cls version 1.8 and later.
% % You must have at least 2 lines in the paragraph with the drop letter
% % (should never be an issue)
% I wish you the best of success.

 
% \hfill December 27, 2012

\IEEEPARstart{C}{ompressive} sensing (CS) \cite{CandesRombergTao:06,
  Candes2006:Optimal, Do} is a recently developed and fast growing field
of research. Given that the signal of interest
is sparse in a certain basis or tight frame, it provides a new sampling scheme that breaks the
conventional Shannon-Nyquist sampling rate
\cite{Shannon1949:samplingtheorem}, which requires sampling at a rate
at least twice the bandwidth of the signal for successful
recovery. % Without loss of generality, let us assume that the signal of
% interest is sparse by itself. Then the CS problem can be
% formulated as follows. For a vector $a \in \mathbb{R}^d$, let
% $\|a\|_0$ denote the $\ell_0$ ``norm'' of $a$, which counts the number of
% nonzero entries in $a$. Then $a$ is \emph{$k$-sparse} if $\|a\|_0
% \le k$. Let $\Phi\in \mathbb{R}^{n\times d}$ ($n \ll d$) be the
% sampling (sensing) matrix. Then $f = \Phi a$ ($f\in \mathbb{R}^n$) is
% called the measurement vector. The goal is to recover the sparse
% signal $a$ from $f$ by solving the
% following $\ell_0$ minimization problem
\my{In its simplest form, compressive sensing addresses the problem of
finding the sparsest problem solution to a set of underdetermined
equations. That is, it addresses the following $\ell_0$ minimization problem
\begin{align}\label{eqn:l0_minimization}
  \min_a \|a\|_0 \mbox{ subject to } f = \Phi a.
\end{align}
where $\|a\|_0$ denotes the $\ell_0$ ``norm'' of $a$, which counts the
number of nonzero elements of $a$, and $\Phi \in \mathbb{R}^{n\times
  d}$ ($n \ll d$). The vector $a$ is said to be $k$-sparse if $\|a\|_0
\le k$.} Candes and Tao~\cite{Candes2005:LP} have established that it
is sufficient to require all submatrices consisting of
arbitrary $2k$ columns of $\Phi$ to have full rank for the $\ell_0$
minimization problem~\eqref{eqn:l0_minimization} to have a unique
$k$-sparse solution. However, finding this solution is, in general, an
NP-hard problem. 

Fortunately, researchers have proposed several approaches to address this
problem, which fall into two main categories. The
first one is to relax the $\ell_0$ minimization problem to an $\ell_1$
minimization
problem~\cite{Candes2005:LP,Candes:08,Foucart2009395,Foucart201097,5290058,5550400,Mo2011460},
which can be solved in polynomial time. The other stream of work is to
use heuristic approaches, such as greedy algorithms, to approximate
the solution of the $\ell_0$ minimization problem~\cite{Gilbert:2003:AFO:644108.644149, DET2, 1337101, DET,
  Needell2009301, 4839056, 5419092, Blumensath2009265,
  Foucart:2011:HTP:2340478.2340494, 6145475}. The analyses of all of these
algorithms depend on properties of the sampling (sensing) matrix
$\Phi$ and two important metrics here are coherence measures and the restricted isometry constant (RIC), both of which are defined below. A useful bridge between these two metrics has been defined in~\cite{YangdeHoog:2014} and %  In these papers, researchers have
% adopted both coherence indices and the restricted isometry constant (RIC) to analyze the performance
% of both streams. In addition, the bridge between coherence indices and
% the RIC has been established in~\cite{YangdeHoog:2014}, and 
used to
study the reconstruction performance of the weak orthogonal matching
pursuit (WOMP) and the orthogonal matching pursuit (OMP).

In this paper, we continue the study of greedy type algorithms, but in a
different flavor. OMP updates an $s$-term approximation of the measurement vector $f$ a step
at a time, adding to an existing $(s-1)$-term set a new term in a
greedy fashion, aiming to
minimize the $\ell_2$ error over all possible combinations of the $s$
terms. However, it is known (see for
instance~\cite{Temlyakov:11}) that the most computationally expensive
step of all greedy algorithms is the greedy step, which calculates in
each iteration the inner products between the residual and all the
atoms from the dictionary and finds the maximum of them. In this
sense, greed is good,  but less greed could be better. 

Here we study a thresholding greedy algorithm called orthogonal matching pursuit with thresholding
(OMPT), which replaces
the expensive greedy step by a thresholding step. It
only needs to calculate the norm of the residual once in each
iteration and uses it as a threshold. We show that by carefully
choosing the thresholding parameter, OMPT is able to recover the $k$
correct support of the ideal signal in presence of noise,
and obtain exact recover of the $k$-sparse signal in noiseless case,
both in $k$ iterations. In addition, by applying the global 2-coherence~\cite{YangdeHoog:2014}, we show that it maintains exactly the same
reconstruction performance as OMP under the same assumptions on
coherence indices and the RIC, for both noisy and noiseless
scenario. % More details are given in the following sections.

\my{The main contributions of this paper can be summarized as
  follows. Greedy algorithms such as OMP and WOMP have been
  studied intensively in the signal processing community. However,
  few thresholding type greedy algorithms have been studied for
  CS. In this paper we analyze the recovery performance and
  convergence of OMPT using the global 2-coherence and the RIC and show
  that OMPT retains exactly the same recovery performance as OMP given
  the same assumption on these two metrics. Specifically, in
Theorem~\ref{thm:ompt_noiseless} and Theorem~\ref{thm:ompt_noisy}, the
recovery properties for OMPT on sparse signals are established for noiseless and noisy cases
respectively, given optimal choices for the threshold
parameter. The convergence of OMPT in presence of noise for
general choice of the threshold parameter
is then analyzed in Theorem~\ref{thm:ompt_convergence}. It is also shown in Corollary~\ref{cor:ompt} that by
carefully choosing the threshold, OMPT has the same reconstruction
performance as OMP. Precisely, the bound on the RIC for OMPT to succeed
is exactly the same as the best known bound for OMP established
in~\cite{YangdeHoog:2014}. As far as we are aware, these results have not been presented in the literature previously.}

\section{Preliminaries and Notation}

Before moving on to the main results of this paper, we need some
preliminaries and notation. Without loss of generality, assume that the columns of the
matrix $\Phi$ are normalized such that for any column
$\phi_i\in\Phi$, $\|\phi_i\|_2 = 1$. We sometimes also refer to the matrix
$\Phi$ as a dictionary, whose columns are called atoms. 

\subsection{Some Notation}
\begin{itemize}
    \item $A_1(\Phi)$: the closure of the convex hull of $\Phi$. Specifically, $A_1(\Phi) = \{g : g = \sum_i
c_i \phi_i, \; \phi_i \in \Phi, \; \sum_i|c_i| \le 1\}$
  \item $\supp(a)$: the support of $a\in\mathbb{R}^d$ is the index set where
the elements of $a$ are nonzero

  \item $|\Lambda|$: the cardinality of the set $\Lambda$

  \item $\Phi_\Lambda$: the sub-dictionary of $\Phi$ with the indices of atoms
restricted to the index set $\Lambda$

  \item $a_\Lambda$: the sub-signal (in $\mathbb{R}^{|\Lambda|}$) of $a\in
    \mathbb{R}^d$ with indices restricted to $\Lambda$
  \item $a_{\text{min}}$: the nonzero element of $a$ with the least magnitude
\end{itemize}

\subsection{Preliminaries}
As mentioned above, there are two types of metrics that are frequently used in the CS
literature, the coherence indices and the RIC. The use of coherence indices can be traced back to~\cite{DH},
where Donoho and Huo used the mutual coherence to describe the
equivalence of the $\ell_0$ minimization and $\ell_1$ minimization.
\begin{defn}\label{def:mutualcoherence}
  The \emph{mutual coherence} $M(\Phi)$ of a matrix $\Phi$ is defined by
  \begin{align*}
    M(\Phi) := \max_{\substack{\phi_i, \phi_j \in \Phi \\ i \neq
        j}} | \langle \phi_i, \phi_j  \rangle |,
  \end{align*}
  where $\langle \cdot, \cdot \rangle$ represents the usual inner product.
\end{defn}
\noindent They showed that~\cite{DH} if $k < (1+M^{-1})/2$, then the $\ell_1$
minimization problem  has a unique solution and is equivalent to
the $\ell_0$ minimization problem. Further results using the mutual
coherence for $\ell_1$ minimization can be found
in~\cite{MR1929464,MR1963681,MR2045813}. Interestingly, OMP shares the
same bound as the $\ell_1$ minimization problem.  It has been shown in~\cite{DET2,DET}
that if $k < (1+M^{-1})/2$, then OMP can recover the true support of the
ideal signal in $k$ iterations in presence of noise, and get exact
recovery of the signal in noiseless case. Moreover, this bound is
known to be sharp.

The other metric, the RIC, was introduced by Candes
and Tao in \cite{Candes2005:LP}.
\begin{defn}[Restricted Isometry Property]
  Let $\Sigma_k$ be the set of $k$-sparse vectors $\Sigma_k = \{u\in\mathbb{R}^d: \|u\|_0 \le k\}$. A matrix $\Phi$ satisfies the \emph{restricted isometry property}
  of order $k$ with the \emph{restricted isometry constant (RIC)} $\delta_k$
  if $\delta_k$ is the smallest constant such that
  \begin{align*}
    (1-\delta_k)\|v\|_2^2 \le \|\Phi v\|_2^2 \le (1+\delta_k)\|v\|_2^2
  \end{align*}
  holds for all $v \in \Sigma_k$.
\end{defn}

% \section{Background}

It is easy to see that the RIC $\delta_k$ increases with $k$ since
$\Sigma_k \subset \Sigma_{k+1}$. Candes shows in \cite{Candes:08} if
$\delta_{2k} < \sqrt{2} - 1$, then $\ell_1$ minimization is equivalent to
$\ell_0$ minimization. Better bounds have been developed
\cite{Foucart2009395, Foucart201097, 5290058, 5550400, Mo2011460} and the most
recent result is $\delta_{2k} < 0.4931$
\cite{Mo2011460}. In contrast to
$\ell_1$ minimization, the conventional metric for a sensing matrix in using greedy
algorithms is usually chosen to be the coherence indices (see for instance
\cite{1337101,DET,5361489}). Recently, researchers have started to investigate the
performance of OMP using the RIC. Davenport and
Wakin~\cite{5550495} have proved that $\delta_{k+1} <
\frac{1}{3\sqrt{k}}$ is sufficient for OMP to recover any $k$-sparse
signal in $k$ iterations. Further improvements have
been reported in \cite{HuangZhu:11, 6092487, 6213142, MoShen:12}. In
particular, \my{Wang and Shim~\cite{6213142}\footnote{\my{In the construction
  of the bound in this paper, the condition was strengthened from
  $\sqrt{k}\delta_{k+1} + \delta_k < 1$ for the first iteration of OMP
  to
  $\delta_{k+1} < \frac{1}{\sqrt{k}+1}$ for the subsequent iterations
  and the main result.}} and} Mo
and Shen~\cite{MoShen:12} have improved the
bound to $\delta_{k+1} < \frac{1}{\sqrt{k}+1}$, and have also given an
example that OMP fails after $k$ iterations when $\delta_{k+1} = \frac{1}{\sqrt{k}}$, as
was conjectured by Dai and Milenkovic in
\cite{4839056}. Zhang\cite{Zhang:2011} has also given a bound
$\delta_{31k} < 1/3$ for OMP to recover a $k$-sparse signal in more
than $k$ iterations. Loosely
speaking, as discussed in~\cite{DuarteEldar:2011}, Zhang's result requires fewer measurements when $k$ is large. However, when $k$ is small, his result is
worse. In addition, Zhang's algorithm requires more than $k$
iterations, and cannot recover the true support of the ideal
sparse signal% when noise is present
.

\my{Whilst coherence measures and the RIC have been used in many studies,
the two metrics have generally been considered
independently. In~\cite{YangdeHoog:2014}, the global 2-coherence was
introduced as a means of providing a bridge between them. }
% While there have been substantial efforts to study these two metrics
% independently, a useful bridge connecting them was missing. One of these
% bridges has been established in~\cite{YangdeHoog:2014}. Let us 
% first introduce the \emph{global 2-coherence}, $\nu_k(\Phi)$ for a given
% dictionary $\Phi$, as defined in~\cite{YangdeHoog:2014}.

\begin{defn}\label{def:newcoherence}
  Denote $[d]$ the index set $\{1,2,\ldots,d\}$. The \emph{global 2-coherence} of a dictionary $\Phi \in \mathbb{R}^{n\times d}$ is defined as
  \begin{equation*}\label{eqn:newcoherence}
    \nu_k(\Phi):=
    \max_{i\in[d]}\max_{\substack{\Lambda\subseteq[d]\setminus\{i\} \\ |\Lambda| \le k}}
    \left( \sum_{j\in\Lambda} \langle \phi_i,\phi_j \rangle^2
    \right)^{1/2},
  \end{equation*}
  where $\phi_i$, $\phi_j$ are atoms from the dictionary $\Phi$.
\end{defn}
\noindent Note that the global 2-coherence $\nu_k (\Phi)$
defined in Definition~\ref{def:newcoherence} is more
general than the mutual coherence defined in
Definition~\ref{def:mutualcoherence}. In fact, when $k=1$, the global
2-coherence defined in Definition~\ref{def:newcoherence} is
exactly the mutual coherence. It is also
more general than the ``local'' 2-coherence function defined
in~\cite{5361489}. 
\my{
  The intuition behind this definition can be seen from the following. For greedy methods which add elements by examining inner products
  of residuals, we require sharp bounds for $|\langle \Phi a, \phi_i
  \rangle|$ where $\|a\|_0 = k < n$ and $\phi_i$ is the $i$th column of
  $\Phi$. Specifically, we need an upper bound for $|\langle \Phi a, \phi_i
  \rangle|$ when $i \notin \Lambda$ and a lower bound for $|\langle \Phi a, \phi_i
  \rangle|$ when $i \in \Lambda$, where $\Lambda = \text{supp}(a)$. Such bounds have been derived
  in~\cite{YangdeHoog:2014}, namely
  \begin{align}
    &\max_{\phi_i, i \notin \Lambda} |\langle \Phi a, \phi_i
    \rangle| \le \nu_k \|a\|_2,  \notag \\
    &\max_{\phi_i, i \in \Lambda} |\langle \Phi a, \phi_i
    \rangle| \ge \frac{(1-\delta_k)\|a\|_2}{\sqrt{k}}. \label{eqn:inner_product_upper_bound}
  \end{align}
  It is straightforward to show that the first bound is sharp in the
  sense that for every sampling matrix $\Phi$ there is a $k$-sparse
  vector $a$ such that the equality holds. Thus, the global 2-coherence is
  a natural metric to use for the analysis of greedy algorithms. Note
  however that the second inequality is not sharp for all choices of
  sampling matrix $\Phi$. We could have used the metric
  \begin{align*}
    \omega_k(\Phi) = \min_{\|x\|_0 = k} \max_{\substack{\phi_i, 
        i\in \Lambda}} \frac{|\langle \Phi x, \phi_i
    \rangle|}{\|x\|_2},
\end{align*}
\mymy{which would replace Equation~\eqref{eqn:inner_product_upper_bound}
with the bound
\begin{align*}
  \max_{\phi_i, i \in \Lambda} |\langle \Phi a, \phi_i
    \rangle| \ge \omega_k(\Phi) \|a\|_2,
  \end{align*}
 which is sharp in the sense that for every sampling matrix $\Phi$
 there is a $k$-sparse vector $a$ such that the equality
 holds. However, we have}
opted to work with the
RIC as it is a more familiar measure and leads to estimates that are
nearly as good. Thus both the RIC and the global 2-coherence are
natural metrics to use in the analysis of greedy algorithms. Note that
the metric $\omega_k(\Phi)$ can be written using operator norm as (see
appendix for details)
\begin{align}\label{eqn:omega}
  \omega_k(\Phi) = \min_{\substack{\Lambda \subset [d] \\ |\Lambda| = k}}
  \frac{1}{\|(\Phi_\Lambda^T \Phi_\Lambda)^{-1}\|_{\infty, 2}},
\end{align}
where $\Phi_\Lambda$ denotes the sub-dictionary of $\Phi$ with indices of atoms
restricted to the index set $\Lambda$.
}\mymy{The mixed norm $\|\cdot\|_{\alpha,\beta}$ here is defined as
  \begin{align*}
    \|A\|_{\alpha,\beta} = \max_{x\neq 0} \frac{\|A\|_\beta}{\|x\|_\alpha}.
\end{align*}
}

\my{The global 2-coherence can also be written as
  \begin{align}
    \nu_k(\Phi) = \max_{\substack{\Lambda\subseteq[d] \\ |\Lambda| \le
        k+1}} \| \Phi_\Lambda^T \Phi_\Lambda - I \|_{\infty, 2},
    \label{eqn:2coherence}
  \end{align}
It is no more complicated than
the RIC, which can be expressed as
\begin{align}
  \delta_k = \max_{\substack{\Lambda\subseteq[d] \\ |\Lambda| \le k}}
  \| \Phi_\Lambda^T \Phi_\Lambda - I \|_{2, 2}.
  \label{eqn:ric}
\end{align}
In fact, from an algorithmic point of view, the global 2-coherence is more useful in
practice as it can be calculated in polynomial time.

Equations~\eqref{eqn:2coherence}~and~\eqref{eqn:ric} are used in the
proof of the following proposition (see~\cite{YangdeHoog:2014} for details).}

\begin{pro}\label{pro:relationship}
    For $k \ge 1$, 
    \begin{equation}\label{eqn:relationship}
        M \le \nu_{k} \le  \delta_{k+1} \le \sqrt{k}\nu_{k} \le kM.
    \end{equation}
  \end{pro}
  \noindent This proposition provides the upper and lower bounds for both the mutual
  coherence $M$ and the RIC $\delta_k$ in terms of the global
  2-coherence $\nu_k$, which establishes the bridge between them. It
  connects the two once independent metrics for greedy algorithms
  together. Moreover, by applying inequalities~\eqref{eqn:relationship}, the authors
  improved the bound on the RIC for OMP to $\delta_k + \sqrt{k}\delta_{k+1}<1$.

Next, in addition to
Proposition~\ref{pro:relationship}, we establish in this paper the relationship among the global 2-coherence $\nu_k$, the cumulative coherence $\mu_{1,k}$ defined in~\cite{1337101},
and the RIC $\delta_k$.

\begin{pro}\label{pro:relationship_new}
    For $k \ge 1$, we have
    \begin{equation*}
        \delta_{k+1} \le \mu_{1,k} \le \sqrt{k}\nu_{k}.
    \end{equation*}
\end{pro}

\begin{proof}
  The inequality $\delta_{k+1} \le \mu_{1,k}$ has been established in
  Proposition 2.10 in~\cite{Rauhut:10}. We next show $\mu_{1,k} \le
  \sqrt{k}\nu_{k}$ for all positive integer $k$.
  \begin{align*}
    \mu_{1,k}
    &=
     \max_{i\in[d]}\max_{\substack{\Lambda\subseteq[d]\setminus\{i\}
         \\ |\Lambda| \le k}} \sum_{j\in\Lambda} |\langle \phi_i,
     \phi_j \rangle| \\
     &\le
      \max_{i\in[d]}\max_{\substack{\Lambda\subseteq[d]\setminus\{i\}
          \\ |\Lambda| \le m}} \left( \sum_{j\in\Lambda} \langle \phi_i,\phi_j \rangle^2
      \right)^{1/2} \cdot \sqrt{k} \\
      &=
      \sqrt{k}\nu_k.
    \end{align*}
  \end{proof}
\noindent  Notice that from the above relations, we see clearly that the cumulative
coherence can only bound the restricted isometry constant from above,
which provides another motivation for introducing the global 2-coherence.

\section{Main Results}\label{sec:main}
In this section, we start the analysis of the recovery properties of
OMPT with the noiseless case, and compare them with the
state-of-art results. We then generalize the results to the case where
a measurement signal is contaminated by a perturbation. A convergence
analysis of OMPT is also given, which can be extended to the more general
Hilbert space. To make the paper more readable, the detailed
proofs of the main results are relegated to the Appendix.

Let us first introduce the OMPT algorithm. This is a thresholding type
modification of OMP and weak OMP (WOMP). It replaces
the expensive greedy step in OMP and WOMP with a thresholding
step. Details are presented in Algorithm~\ref{alg:ompt}, where
$\Phi_{\Lambda_s}$ denotes the sub-dictionary of $\Phi$ with atoms
restricted to the index set $\Lambda_s$ from the $s$-th iteration, and $\hat{a}_{\Lambda_k}$
denotes $\hat{a}$ restricted to the support set $\Lambda_k$ after $k$
iterations. \myIII{Note that in the thresholding step (step 4 in
  Algorithm~\ref{alg:ompt}), unlike some multi-index thresholding
  algorithms such as StOMP~\cite{6145475}, only the first index satisfying the
  thresholding condition is picked from a randomly permuted index set.} An
initial study of this algorithm in Hilbert space was presented
in~\cite{Yang:2011:GAA:2338237}. For the case of Banach space, a
similar version of the algorithm was studied in~\cite{Temlyakov:08}
(see also \cite{T9} and \cite{Temlyakov:11}). However, as far as we
are aware, the present paper provides  the first analysis of the
performance of this algorithm for sparse signal recovery in
CS using different coherence indices and the restricted isometry
constant.

\begin{algorithm}
  \caption{Orthogonal Matching Pursuit with Thresholding (OMPT)}
  \begin{algorithmic}[1]
    \STATE\textbf{Input:} threshold $t$, dictionary $\Phi$,
    signal $f$.
    \STATE\textbf{Initialization:} $r_0 := f$, $x_0:=0$, 
    $\Lambda_0:=\emptyset$, $s:=0$.
    \WHILE {$\|r_s\|_2 > t \|f\|_2$}
      \STATE Find an index $i$ such that $$|\langle r_s,\phi_i\rangle|
      \ge t\|r_s\|_2;$$
      \STATE Update the support: $$\Lambda_{s+1} = \Lambda_s \cup \{i\};$$
      \STATE Update the estimate: $$x_{s+1} = \argmin_z \|f - \Phi_{\Lambda_{s+1}} z\|_2;$$
      \STATE Update the residual: $$r_{s+1} = f - \Phi_{\Lambda_{s+1}} x_{s+1};$$
      \STATE $s = s+1$;
    \ENDWHILE
    \STATE \textbf{Output:} If the algorithm is stopped after $k$
    iterations, then the output estimate $\hat{a}$ of $a$ is
    $\hat{a}_{\Lambda_k} = x_k$ and $\hat{a}_{\Lambda_k^C} = 0$.
  \end{algorithmic}
  \label{alg:ompt}
\end{algorithm}

First we study the recovery properties of the OMPT algorithm. We start
with the ideal noiseless case where the measurement signal is
obtained by encoding a sparse signal. Specifically, let
$\Lambda\subset [d]$ with $|\Lambda| = k$. We consider a measurement
$f = \Phi a$, where $\Phi \in
\mathbb{R}^{n\times d}$ and $a\in\mathbb{R}^d$ with $\supp(a) = \Lambda$.
We have the following results.

\begin{thm}\label{thm:ompt_noiseless}
  Let $f = \Phi a$ with $\|a\|_0 = k$. If
  \begin{align}
    \delta_k + \sqrt{k}\nu_k < 1 \label{eqn:thm_ompt_bound}
  \end{align}
  and
  \begin{align}
    \frac{\nu_k}{\sqrt{1-\delta_k}} < t \le
    \frac{\sqrt{1-\delta_k}}{\sqrt{k}}, \label{eqn:ompt_thm_condition}
  \end{align}
  then $a$ is the unique sparsest representation of $f$ and moreover,
  OMPT recovers $a$ exactly in $k$ iterations.
\end{thm}
\my{\noindent This result for noiseless recovery property of OMPT involves the
  global 2-coherence and the RIC. This is the most accurate bound in
  the paper and is the basis for the derivation of simpler bounds
  given in Corollary~\ref{cor:ompt}. As the proof of this theorem is a special case of
Theorem~\ref{thm:ompt_noisy}, we refer the reader to the proof of Theorem~\ref{thm:ompt_noisy}, which is given in the appendix.}

% \my{\noindent This is the best result we are able the get in this
%   paper for the noiseless
% recovery property of OMPT, which is a hybrid of the global-2
% coherence and the RIC. Everything else follows from this result and Proposition~\ref{pro:relationship} and Proposition~\ref{pro:relationship_new}.
% The proof of this theorem is a special case of
% Theorem~\ref{thm:ompt_noisy} and is therefore not given in the
% appendices.}

% \section{Discussion}\label{sec:discussion}
Notice that the condition~\eqref{eqn:ompt_thm_condition} on the
threshold $t$ requires the
bound~\eqref{eqn:thm_ompt_bound} as a sufficient
condition. Next we will examine this bound and compare it with the
corresponding bound for OMP by exploring the relationships among
different coherence indices and the restricted isometry
constant.

Now by applying Proposition~\ref{pro:relationship} and Proposition~\ref{pro:relationship_new} to
Theorem~\ref{thm:ompt_noiseless}, we obtain the following corollary.

\begin{cor}\label{cor:ompt}
  Let $f = \Phi a$ with $\|a\|_0 = k$ ($k \ge 2$). If any of the following four
  conditions is satisfied:
  \begin{enumerate}[i)]
    \item
      \begin{align}
        \delta_k + \sqrt{k} \delta_{k+1} <
        1 \label{eqn:ompt_cor_condition_delta} 
      \end{align}
      and
      \begin{align*}
        \frac{\delta_{k+1}}{\sqrt{1-\delta_k}} < t \le
        \frac{\sqrt{1-\delta_k}}{\sqrt{k}},
      \end{align*}
    \item
      \begin{align*}
        \nu_k\sqrt{k} + \nu_{k-1}\sqrt{k-1} < 1 % \label{eqn:ompt_cor_condition_nu}
      \end{align*}
      and
      \begin{align*}
        \frac{\nu_k}{\sqrt{1-\nu_{k-1}\sqrt{k-1}}} < t \le
        \frac{\sqrt{1-\nu_{k-1}\sqrt{k-1}}}{\sqrt{k}}, \notag
      \end{align*}
    \item
      \begin{align*}
        \mu_{1,k-1} + \sqrt{k}\mu_{1,k} < 1 % \label{eqn:ompt_cor_condition_mu}
      \end{align*}
      and
      \begin{align*}
        \frac{\mu_{1,k}}{\sqrt{1-\mu_{1,k-1}}} < t \le
        \frac{\sqrt{1-\mu_{1,k-1}}}{\sqrt{k}}, \notag
      \end{align*}
    \item
      \begin{align}
        M < \frac{1}{2k-1} \label{eqn:ompt_cor_condition_M}
      \end{align}
      and
      \begin{align}
        \frac{\sqrt{k}M}{\sqrt{1-(k-1)M}} < t \le
        \frac{\sqrt{1-(k-1)M}}{\sqrt{k}}, \label{eqn:ompt_cor_condition_t}
      \end{align}
    \end{enumerate}
  then $a$ is the unique sparsest representation of $f$ and moreover,
  OMPT recovers $a$ exactly in $k$ iterations.
\end{cor}

As we can see from the above corollary, although we are replacing the most difficult
  (expensive) step of OMP,
  namely the greedy step, by a very simple thresholding step, thus
  making it substantially more efficient, there is no performance degrading. The
  bound~\eqref{eqn:ompt_cor_condition_delta} on the restricted
  isometry constant $\delta_k$ is exactly the same as
  the bound in\cite[Corollary 3.3]{YangdeHoog:2014}
  for OMP. The bound~\eqref{eqn:ompt_cor_condition_M} on
  mutual coherence $M$ also coincides with the
  best known bound~\cite{Temlyakov:11}. Notice that the
  bound~\eqref{eqn:ompt_cor_condition_delta} gives an improved bound
  on the restricted isometry constant compared to the bound obtained
  in\cite{6213142,MoShen:12} for exact recovery of a $k$-sparse signal in
  $k$ iterations, where the bound was $\delta_{k+1} < \frac{1}{\sqrt{k}+1}$.

Theorem~\ref{thm:ompt_noiseless} can be generalized to the
noisy case. Specifically, Let $\Lambda \subset
[d]$ with $|\Lambda| = k$. We consider a measurement $f = \Phi
a + w$, where $a\in\mathbb{R}^d$ with $\supp(a) = \Lambda$ and
$\|w\|_2 \le \epsilon$. We will inspect the recovery performance of OMPT
after $k$ iterations. Denote by $a_{\text{min}}$ the nonzero entry of the
sparse signal $a$ with the least
  magnitude. The following two lemmas will be needed.

\begin{lem}\label{lem:t_lowerbound_noisy}
  Consider the residual at the $s$-th iteration of OMPT $r_s = \Phi
  a_s + w_s$.
  If
  % \begin{align}
  %   t >
  %   \frac{\nu_k\|a\|_\infty+\epsilon}{\sqrt{1-\delta_k}\|a\|_\infty-\epsilon}, \label{eqn:t_lowerbound_noisy}
  % \end{align}
  \begin{align}
    t >
    \frac{\nu_k |a_{\emph{min}}|+\epsilon}{\sqrt{1-\delta_k}
      |a_{\emph{min}}|-\epsilon}, \label{eqn:t_lowerbound_noisy}
  \end{align}
  then
  \begin{align*}
    \max_{i \in [d]\setminus\Lambda} |\langle r_s,\phi_i \rangle| < t \|r_s\|_2.
  \end{align*}
\end{lem}

\begin{lem}\label{lem:t_upperbound_noisy}
 Consider the residual at the $s$-th iteration of OMPT $r_s = \Phi
  a_s + w_s$. Given
  \begin{align}
    t \le \frac{(1-\delta_k) |a_{\emph{min}}| -
      \epsilon}{\sqrt{k(1-\delta_k)} |a_{\emph{min}}|+\epsilon}, \label{eqn:t_upperbound_noisy}
  \end{align}
  we have
  \begin{align*}
    \max_{i \in \Lambda} |\langle r_s,\phi_i \rangle| \ge t \|r_s\|_2.
  \end{align*}
\end{lem}

\begin{thm}\label{thm:ompt_noisy}
   Denote by $\hat{a}_{\emph{ompt}}$ the recovered signal from
  $f$ by OMPT after $k$ iterations. If
  \begin{align}
    \delta_k + \sqrt{k}\nu_k < 1 \label{eqn:thm_ompt_bound1_noisy}
  \end{align}
  and the noise level obeys
  % \begin{align}
  %   |a_{\text{min}}| > \frac{\epsilon\Big((1-\delta_k) +
  %     k\sqrt{k(1+\delta_k)}\Big)}{(1-\delta_k) \sqrt{k(1+\delta_k)} - k\nu_k\sqrt{1+\delta_k}},
  % \end{align}
  % and the threshold $t$ satisfies
  % \begin{align}
  %   &\frac{\nu_k}{\sqrt{1-\delta_k}} \left( 1 +
  %     \frac{\epsilon \sqrt{k}}{|a_{\text{min}}| \nu_k} \right) < t \notag\\
  %   &\le \frac{\sqrt{1-\delta_k}}{\sqrt{k}} \left( 1 -
  %     \frac{\epsilon}{|a_{\text{min}}| \sqrt{k(1+\delta_k)}} \right),
  % \end{align}
  \begin{align}
    \epsilon < \frac{\sqrt{1-\delta_k} (1-\delta_k -
      \sqrt{k}\nu_k)}{(\sqrt{k}+1)\sqrt{1-\delta_k} +      
       (1-\delta_k) + \nu_k}|a_{\emph{min}}|, \label{eqn:thm_ompt_bound2_noisy}
    % |a_{\text{min}}| > \frac{(\sqrt{k}+1)\sqrt{1-\delta_k} +
    %   (1-\delta_k) + \sqrt{k}\nu_k}{\sqrt{1-\delta_k} (1-\delta_k -
    %   \sqrt{k}\nu_k)} \epsilon, \label{eqn:thm_ompt_bound2_noisy}
  \end{align}
  then there exists threshold $t$ satisfying conditions
  \eqref{eqn:t_lowerbound_noisy} and \eqref{eqn:t_upperbound_noisy}.
  Moreover, we have
  \begin{enumerate}[a)]
        \item $\hat{a}_{\emph{ompt}}$ has the correct sparsity pattern, that is
            \begin{equation*}
                \supp(\hat{a}_{\emph{ompt}}) = \supp(a);
            \end{equation*}
         \item $\hat{a}_{\emph{ompt}}$ approximates the ideal
           noiseless representation and
             \begin{equation}
                 \|\hat{a}_{\emph{ompt}} - a\|_2^2 \leq
                 \frac{\epsilon^2}{1-\delta_k}. \label{eqn:thm_ompt_errorbound_noisy}
             \end{equation}
    \end{enumerate}
\end{thm}

Theorem~\ref{thm:ompt_noisy} basically says that, if the minimal nonzero
entry of the ideal noiseless sparse signal is large enough
compared to the noise level, then the correct support of the sparse
signal can be recovered exactly in $k$ iterations, and moreover, the error can be bounded by \eqref{eqn:thm_ompt_errorbound_noisy}. 

Next we study the convergence of the OMPT algorithm in presence of noise.

% \begin{defn}
%   Given a dictionary $\mathcal{D}$, the convex hull of $\mathcal{D}$
%   is defined as $A_1(\mathcal{D}) = \{g: g = \sum_i c_i \phi_i, \;
%   \phi_i \in \mathcal{D}, \; \sum_i |c_i| \le 1\}$.
% \end{defn}

\begin{thm}\label{thm:ompt_convergence}
  Given a dictionary $\Phi$, take $\epsilon \ge 0$ and $f, f^\epsilon
  \in \mathbb{R}^n$ such that
  \begin{align*}
    \|f - f^\epsilon\| \le \epsilon, \quad\quad f^\epsilon /
    C(\epsilon) \in A_1(\Phi)
  \end{align*}
  with some constant $C(\epsilon) > 0$. Then OMPT stops after $m \le 
  \ln t^2 / \ln(1 - t^2)$ iterations with
  \begin{align*}
    \|r_m\| \le \epsilon + t C(\epsilon).
  \end{align*}
\end{thm}

\begin{rmk}
Note that from the proof of Theorem~\ref{thm:ompt_convergence} (see
Appendix), it is clear that the residual term $r_m$ converges in an exponential rate. In addition, the
results in Theorem~\ref{thm:ompt_convergence} can be easily extended
to any Hilbert space.
\end{rmk}
\my{
\begin{rmk}
  Also note that when we choose $t^2$ to be close to
  $1/k$ as in Lemma~\ref{lem:t_upperbound_noisy}, the bound $\ln t^2 / \ln (1-t^2)$ for the number of
  iterations $m$ is roughly $k \ln k$. This can be seen by using the fact that $(1 - 1/k)^k$ is approximately $1/e$.
\end{rmk}
}

%\caption{An Example of a Table}
%\label{table_example}
%\centering
%% Some packages, such as MDW tools, offer better commands for making tables
%% than the plain LaTeX2e tabular which is used here.
%\begin{tabular}{|c||c|}
%\hline
%One & Two\\
%\hline
%Three & Four\\
%\hline
%\end{tabular}
%\end{table}

% Note that IEEE does not put floats in the very first column - or typically
% anywhere on the first page for that matter. Also, in-text middle ("here")
% positioning is not used. Most IEEE journals use top floats exclusively.
% Note that, LaTeX2e, unlike IEEE journals, places footnotes above bottom
% floats. This can be corrected via the \fnbelowfloat command of the
% stfloats package.

\mymy{
\section{Simulations}
In this section, we present simulations that compare the
reconstruction performance and the complexity of OMPT with
OMP. For this purpose, we use similar setup to that given in~\cite{DET}. Specifically, we work
with a dictionary $\Phi = [I, F] \in \mathbb{R}^{128 \times 256}$, concatenating a standard basis and a
Fourier basis for signals of length $128$ together. The sparse signals we
test are obtained by randomly choosing locations for
nonzero entries, and then assigning values from the uniform
distribution to these locations. We perform 1000 trials for each of such
sparse signals generated with certain sparsity level.

We first examine the reconstruction performance of OMPT and compare it
with that of OMP. For simplicity, we only consider the metric
using mutual coherence $M$. In particular, we choose for OMPT the
thresholding parameter $t = \sqrt{M} = 1/\sqrt{128}$, which satisfies the
condition~\eqref{eqn:ompt_cor_condition_t} for sparsity $k = 1,
\ldots, 6$ satisfying
condition~\eqref{eqn:ompt_cor_condition_M}. Figure~\ref{fig:performance}
shows the reconstruction performance of OMPT under the above settings
and that of its counter part OMP. The $x$-axis is the sparsity level and the
$y$-axis is the expected value of success. We see from the figure that OMPT has
similar performance as OMP. In particular, OMPT performs better than
OMP when sparsity $k$ is less than 30. Moreover, OMPT does not fail
when $k$ is no more than 20.

\begin{figure}[ht!]
  \centering
  \includegraphics[scale=.4]{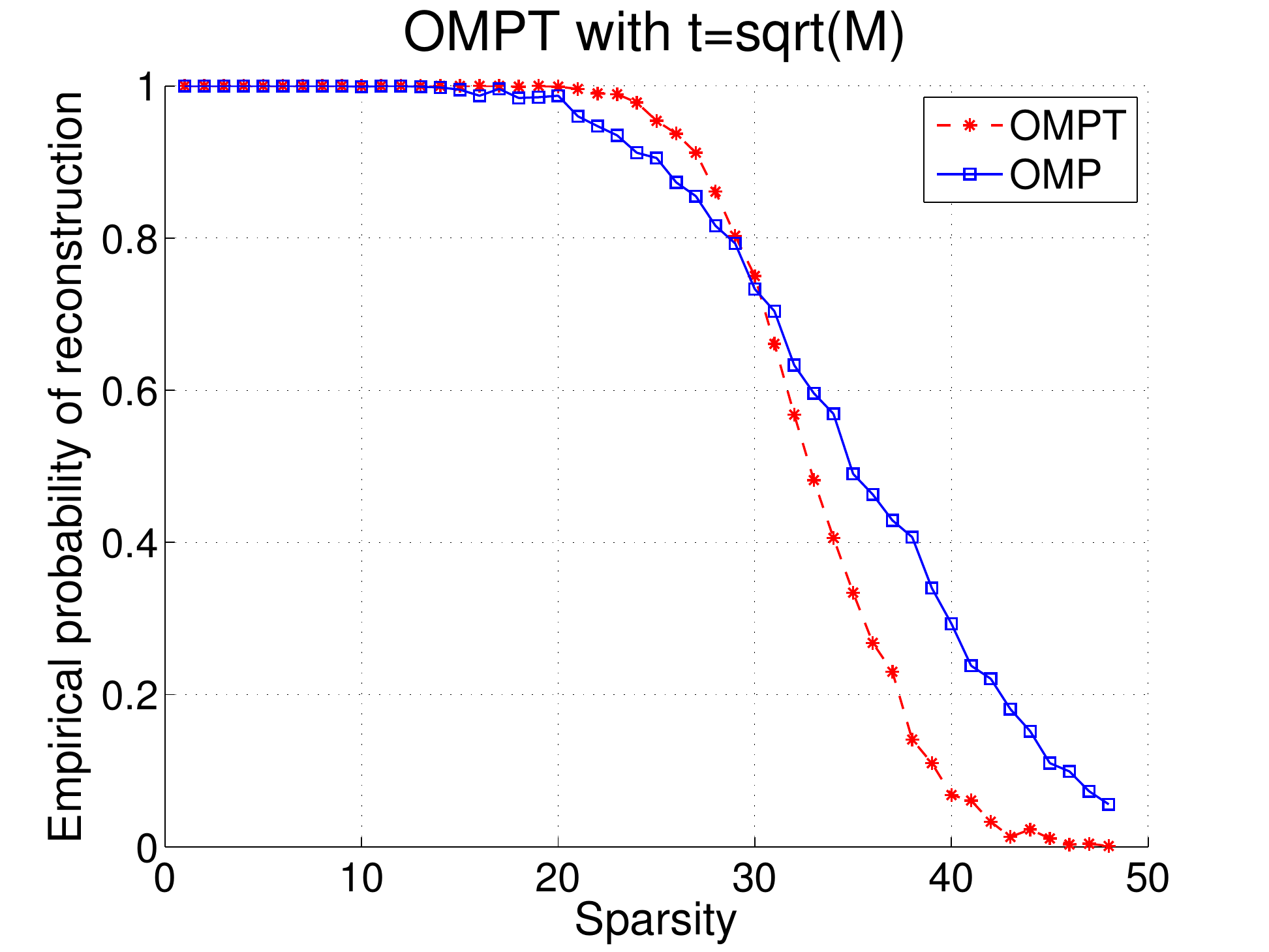}
  \caption{Average reconstruction performance of OMPT and OMP for
    sparse signals of length 256 in
    1000 trials.}
  \label{fig:performance}
\end{figure}

We next compare the complexity of OMPT and OMP in terms of number of
inner products needed. For OMP, one can actually calculate the number
of inner products, which equals $k(2d-k+1)/2$. It increases as
sparsity $k$ increases, as shown in Figure~\ref{fig:complexity}. In
contrast, we count the number of inner products performed by OMPT for
1000 trials. The average is shown in
Figure~\ref{fig:complexity}. As one can see, it stays flat as $k$
increases and is significantly less than the number of inner products needed
for OMP.

\begin{figure}[ht!]
  \centering
  \includegraphics[scale=.4]{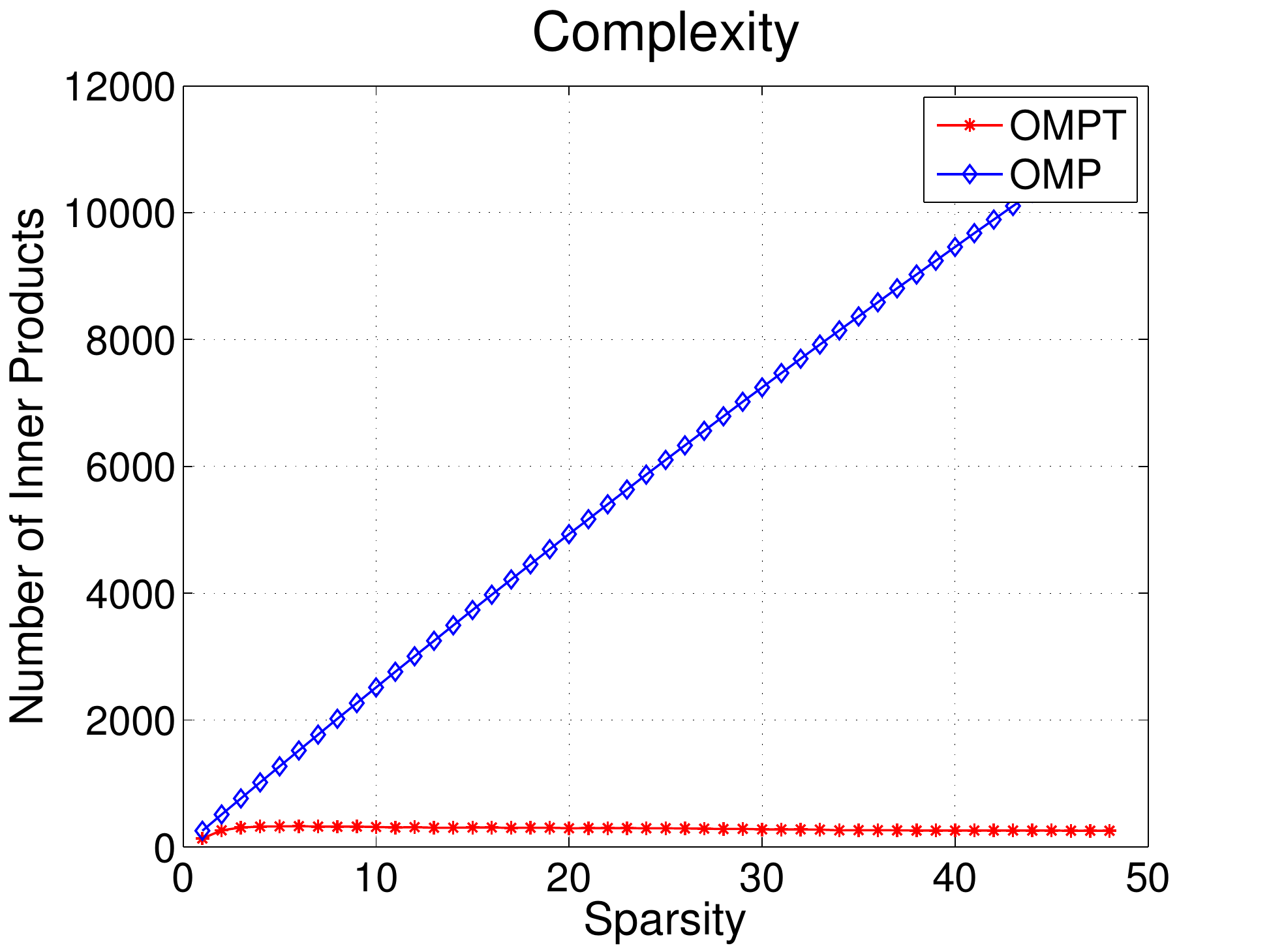}
  \caption{Average number of inner products performed by OMPT and OMP
    in 1000 trials.}
  \label{fig:complexity}
\end{figure}

}

\section{Conclusion}
In this paper, we have analyzed a less greedy algorithm, orthogonal matching pursuit with
thresholding and its performance for reconstructing sparse signals,
for both noisy and noiseless cases. It
replaces the expensive greedy step in orthogonal matching pursuit with
a thresholding step, making it a potentially  more attractive option
in practice. By analysing different metrics for the sampling matrix
such as the RIC, mutual coherence and global 2-coherence, we showed
that although the expensive greedy step is replaced, this simplified
algorithm has the same recovery performance as orthogonal matching
pursuit for sparse signal reconstruction.

% or


% \appendices

% % you can choose not to have a title for an appendix
% % if you want by leaving the argument blank
% \section{}

\appendices

\section{Proof of expression~\eqref{eqn:omega}}
\begin{proof}
  \begin{align*}
    \max_{\phi_i, i\in S} \frac{|\langle \Phi x, \phi_i
    \rangle|}{\|x\|_2}
    &=
    \frac{\|\Phi_S^T \Phi_S x_S\|_\infty}{\|x_S\|_2} \\
    &=
    \frac{\|\Phi_S^T \Phi_S x_S\|_\infty}{\|(\Phi_S^T
      \Phi_S)^{-1}\Phi_S^T \Phi_S x_S\|_2} \\
    &\ge
    \frac{1}{\|(\Phi_S^T \Phi_S)^{-1}\|_{\infty, 2}}.
  \end{align*}
  Therefore,
  \begin{align*}
    \omega_k(\Phi) = \min_{\substack{S \subset [d] \\ |S|=k}}
    \frac{1}{\|(\Phi_S^T \Phi_S)^{-1}\|_{\infty, 2}}.
  \end{align*}
\end{proof}

\section{Proof of Theorem~\ref{thm:ompt_noisy}}

The proof of Lemma~\ref{lem:t_lowerbound_noisy}
and Lemma~\ref{lem:t_upperbound_noisy} will need the following lemma.

\begin{lem}\label{lem:bounds_noisy}
  Let $\Lambda \subset [d]$ with $|\Lambda| = k$. Let $f
  = \Phi a + w$ with $\supp(a) = \Lambda$ and $\|w\|_2 \le
  \epsilon$. In addition, assume that there exits $\Omega \subseteq
  \Lambda$ with $|\Omega| = m$, such that
  \begin{align*}
    \langle \Phi a, \phi_i \rangle = 0, \mbox{ for $i\in \Lambda \setminus
      \Omega$}.
  \end{align*}
  Then
  \begin{align*}
    \max_{i \in [d]\setminus\Lambda} |\langle f,\phi_i \rangle|
    &\le
    \nu_k \|a\|_2 + \epsilon, \\
    \max_{i \in \Lambda} |\langle f,\phi_i \rangle|
    &\ge
    \frac{\sqrt{1-\delta_k}}{\sqrt{m}}  \|\Phi a\|_2 - \epsilon.
  \end{align*}
\end{lem}

\begin{proof}
  For $i \in [d]\setminus\Lambda$, we have
  \begin{align*}
    |\langle f,\phi_i \rangle|
    &=
    | \langle \Phi a + w, \phi_i \rangle| \\
    &\le
   |\langle \Phi a,\phi_i \rangle| + |\langle w,\phi_i \rangle|\\
   &\le
   \left| \sum_{j\in\Lambda} a_j \langle \phi_j, \phi_i \rangle
   \right| + \|w\|_2 \|\phi_i\|_2\\
   &\le
   \|a\|_2 \left( \sum_{j\in\Lambda} |\langle \phi_j, \phi_i
     \rangle|^2 \right)^{\frac{1}{2}} + \epsilon\\
   &\le
   \nu_k \|a\|_2 + \epsilon.
 \end{align*}
	Taking maximum on both sides completes the proof of the first inequality.

	Now for $i \in \Lambda$, we have
        \begin{align*}
          \|\Phi a\|_2^2
          &=
          \langle \Phi a, \sum_{i\in\Lambda} a_i \phi_i \rangle \\
          &=
          \sum_{i\in\Lambda} a_i \langle \Phi a, \phi_i \rangle \\
          &=
          \sum_{i\in\Omega} a_i \langle \Phi a, \phi_i \rangle \\
          &\le
          \sum_{i \in \Omega} |a_i|\cdot |\langle \Phi a,\phi_i \rangle| \\
          &\le
          \sqrt{m} \|a\|_2 \max_{i \in \Lambda} |\langle \Phi a,\phi_i \rangle|.
        \end{align*}
        By the definition of the RIC $\delta_k$,
        \begin{align*}
          \max_{i \in \Lambda} |\langle \Phi a , \phi_i \rangle| \ge \frac{\sqrt{1-\delta_k}}{\sqrt{m}} \|\Phi a\|_2.
        \end{align*}
        Therefore,
	\begin{align*}
          \max_{i \in \Lambda} |\langle f,\phi_i \rangle|
          &=
          \max_{i \in \Lambda} |\langle \Phi a + w,\phi_i \rangle|  \\
          &\ge
          \max_{i \in \Lambda} |\langle \Phi a , \phi_i \rangle| -
          \max_{i \in \Lambda} |\langle w, \phi_i \rangle| \\
          &\ge
          \frac{\sqrt{1-\delta_k}}{\sqrt{m}} \|\Phi a\|_2 - \max_{i
            \in \Lambda} \|w\|_2 \|\phi_i\|_2 \\
          &\ge
          \frac{\sqrt{1-\delta_k}}{\sqrt{m}} \|\Phi a\|_2 - \epsilon.
        \end{align*}
        This completes the proof of the second inequality.
\end{proof}

Now we are ready to complete the proof of
Lemma~\ref{lem:t_lowerbound_noisy},
Lemma~\ref{lem:t_upperbound_noisy}, and Theorem~\ref{thm:ompt_noisy}.

\begin{proof}[Proof of Lemma~\ref{lem:t_lowerbound_noisy}]
  By using Lemma~\ref{lem:bounds_noisy}, assumption~\eqref{eqn:t_lowerbound_noisy}, and the fact that
  $\|a_s\|_2 \ge \sqrt{k-s} |a_{\text{min}}|$, it is easy to derive
  \begin{align*}
    \max_{i \in [d]\setminus\Lambda} |\langle r_s,\phi_i \rangle|
    &\le
    \nu_k \|a_s\|_2 + \epsilon \\
    &<
    t \left( \sqrt{1-\delta_k}\|a_s\|_2 - \epsilon \right) \\
    &\le
    t \left( \|\Phi a_s\|_2 -\epsilon \right) \\
    &\le
    t \|r_s\|_2.
  \end{align*}
\end{proof}

\begin{proof}[Proof of Lemma~\ref{lem:t_upperbound_noisy}]
  By using Lemma~\ref{lem:bounds_noisy}, assumption~\eqref{eqn:t_upperbound_noisy}
  and the inequality
  \begin{align*}
    \sqrt{(k-s)(1-\delta_k)} |a_{\text{min}}| \le \sqrt{1-\delta_k}\|a\|_2 \le
    \|\Phi a\|_2
  \end{align*}
  we have
  \begin{align*}
    \max_{i \in \Lambda} |\langle r_s,\phi_i \rangle| 
    &\ge
    \frac{\sqrt{1-\delta_k}}{\sqrt{k-s}} \|\Phi a_s\|_2 - \epsilon \\
    &\ge
    t (\|\Phi a_s\|_2 + \epsilon) \\
    &\ge
    t \|r_s\|_2.
  \end{align*}
\end{proof}

\begin{proof}[Proof of Theorem~\ref{thm:ompt_noisy}]
  First we show that $a_{\text{ompt}}$ has the correct support.

  We start with the first iteration. Combining conditions
  \eqref{eqn:thm_ompt_bound1_noisy} and
  \eqref{eqn:thm_ompt_bound2_noisy},
  Lemma~\ref{lem:t_lowerbound_noisy} and
  Lemma~\ref{lem:t_upperbound_noisy}, it is easy to see that OMPT is
  able to select and only select an atom $\phi_i$ with $i \in
  \Lambda$. Condition~\eqref{eqn:thm_ompt_bound1_noisy} and
  ~\eqref{eqn:thm_ompt_bound2_noisy} guarantees the existance of
  threshold $t$ satisfying Lemma~\ref{lem:t_lowerbound_noisy} and
  Lemma~\ref{lem:t_upperbound_noisy}. Lemma~\ref{lem:t_lowerbound_noisy}
  guarantees that OMPT will not choose any atom $\phi_i$ for $i \in
  [d]\setminus\Lambda$. While Lemma~\ref{lem:t_upperbound_noisy} guarantees that
  OMPT is able to choose atoms $\phi_i$ with $i \in \Lambda$.

  Next we argue that by repeatedly applying
  Lemma~\ref{lem:t_lowerbound_noisy} and
  Lemma~\ref{lem:t_upperbound_noisy}, we are
  able to correctly recover
  the support of $a$. In fact, in each iteration, we have the same
  situation as in the first iteration. In addition, the orthogonal
  projection step guarantees that the procedure will not repeat the
  atoms already chosen in previous iterations. Thus, all the correct
  support of the noiseless sparse signal $a$ can be recovered
  precisely after $k$ iterations.

  Next, we will prove the error
  bound~\eqref{eqn:thm_ompt_errorbound_noisy}. The proof follows the
  idea of Theorem 5.1 in \cite{DET}. Let $a_T$ denote $a$ restricted
  to its support. Similarly, let $\Phi_T$ denote the dictionary $\Phi$
  restricted to the support of $a$. The orthogonal projection step
  tells that OMPT solves for
  \begin{align*}
    \hat{a}_T = \argmin_{a_T} \| f - \Phi_T a_T\|_2 = \Phi_T^\dagger f
  \end{align*}
  where $\Phi_T^\dagger$ denotes the Moore-Penrose generalized inverse
  of $\Phi_T$. Then we have
  \begin{align*}
    \hat{a}_T
    &=
    \Phi_T^\dagger f \\
    &=
    \Phi_T^\dagger (\Phi a + w) \\
    &=
    \Phi_T^\dagger (\Phi_T a_T + w) \\
    &=
    a_T + \Phi_T^\dagger w.
  \end{align*}
  The term $\Phi_T^\dagger w$ denotes the reconstruction error. It can
  be bounded by
  \begin{align*}
    \|\hat{a}_{\text{ompt}} - a\|_2
    &=
    \|\hat{a}_T - a_T\|_2 \\
    &=
    \|\Phi_T^\dagger w\|_2 \\
    &\le
    \|\Phi_T^\dagger\|_2 \cdot \|w\|_2 \\
    &\le
    \epsilon/\sigma_{\text{min}}
  \end{align*}
  where we bound the norm of $\Phi_T^\dagger$ by the smallest singular
  value $\sigma_{\text{min}}$ of $\Phi$. Now by RIP, we have $\sigma_{\text{min}}^2
  \ge 1 - \delta_k$, and the error
  bound~\eqref{eqn:thm_ompt_errorbound_noisy} follows.
\end{proof}

  \section{Proof of Theorem~\ref{thm:ompt_convergence}}
  \begin{proof}[Proof of Theorem~\ref{thm:ompt_convergence}]
  If the stopping criteria $\|r_m\| \le t \|f\|$ is met, then the
  error estimation follow from the simple inequalities
  \begin{align*}
    \|f\| \le \|f - f^\epsilon\| + \|f^\epsilon\| \le \epsilon + C(\epsilon).
  \end{align*}
  Now assume the stopping criteria $|\langle r_m, \phi
  \rangle| < t \|r_m\|$ has been met for all $\phi \in \Phi$, and
  denote $G_m$ the approximant after $m$ iterations. Then
  \begin{align*}
    \|r_m\|^2
    &=
    \langle r_m, f - G_m \rangle = \langle r_m, f\rangle \\
    &=
    \langle r_m, f - f^\epsilon \rangle + \langle r_m, f^\epsilon
    \rangle \\
    &\le
    \|r_m\|\cdot \|f-f^\epsilon\| + C(\epsilon) \langle r_m,
    \sum_{\phi_i\in\Phi} a_i \phi_i \rangle \\
    &\le
    \epsilon \|r_m\| + C(\epsilon) \sum_{\phi_i\in\Phi} a_i
    \langle r_m, \phi_i \rangle \\
    &\le
    \epsilon \|r_m\| + t C(\epsilon) \|r_m\| \sum_{\phi_i\in\Phi} |a_i|\\
    &\le
    \epsilon \|r_m\| + t C(\epsilon)\|r_m\|.
  \end{align*}
  Therefore, we obtain the bound
  \begin{align*}
    \|r_m\| \le \epsilon + t C(\epsilon).
  \end{align*}
  Next, we prove the bound on the number of iterations. Suppose we are
  at the $k$-th iteration and have found $\phi_i$ such that
  \begin{align*}
    |\langle \phi_i, r_{k-1} \rangle| \ge t\|r_{k-1}\|.
  \end{align*}
  We now update the support
  \begin{align*}
    \Lambda_k = \Lambda_{k-1} \cup \{i\}.
  \end{align*}
  and calculate the new approximant
  \begin{align*}
    G_k = \Phi_{\Lambda_k} \cdot \arg\min_z \|f - \Phi_{\Lambda_k} z\|^2.
  \end{align*}
  Then
  \begin{align*}
    \|r_k\|^2
    &=
    \min_z \|f - \Phi_{\Lambda_k} z\|^2 \\
    &\le
    \|r_{k-1} - \langle r_{k-1}, \phi_i \rangle \phi_i \|^2 \\
    &=
    \|r_{k-1}\|^2 - \langle r_{k-1}, \phi_i \rangle^2 \\
    &\le
    \|r_{k-1}\|^2 - t^2 \|r_{k-1}\|^2 \\
    &=
    (1-t^2) \|r_{k-1}\|^2.
  \end{align*}
  Hence
  \begin{align*}
    \|r_k\| \le (1-t^2)^{\frac{1}{2}} \|r_{k-1}\|,
  \end{align*}
  which implies
  \begin{align*}
    \|r_k\| \le (1-t^2)^{\frac{k}{2}} \|f\|. % \label{eqn:ompt_convergence}
  \end{align*}
  
    From the stopping criteria $\|r_m\| \le t \|f\|$, we know that the
    algorithm will stop after $m$ iterations where $m$ is the smallest
    integer satisfying
    \begin{align*}
      \left( 1 - t^2 \right)^{\frac{k}{2}} \le t.
    \end{align*}
    This implies
    \begin{align*}
      m \le \frac{2\ln t}{\ln(1 - t^2)}.
    \end{align*}
  \end{proof}

% % use section* for acknowledgement
% \my{
% \section*{Acknowledgment}

% The authors would like to thank the anonymous reviewers for their valuable comments
% and suggestions.
% }

\ifCLASSOPTIONcaptionsoff
  \newpage
\fi



\bibliographystyle{IEEEtran}
\bibliography{ga}
\end{document}